%% file: final_draft.tex
\PassOptionsToPackage{table}{xcolor} 
\documentclass[sigconf, nonacm]{acmart}

\usepackage{xspace}
\usepackage{tabularx}
\usepackage{enumitem}
\usepackage{mathtools}
\usepackage{comment}
\usepackage{amsmath}
\usepackage{amsthm}
\usepackage{amsfonts}
\usepackage{graphicx}
\usepackage{subcaption}
\usepackage{etoolbox}
\usepackage{booktabs} 
\usepackage{hyperref}
\usepackage{algorithm, float, setspace}
\usepackage{algorithmic}
\usepackage{relsize}
\usepackage{natbib}
\usepackage{color}
\usepackage{makecell}

\definecolor{darkblue}{rgb}{0.0,0.0,0.5}
\hypersetup{colorlinks=true,breaklinks=true,bookmarks=true,urlcolor=darkblue,citecolor=darkblue,linkcolor=darkblue,bookmarksopen=false,draft=false}

\usepackage{xcolor} 
\newcommand\xtw[1]{\textcolor{blue}{[xw: #1]}}
\newcommand\gs[1]{\textcolor{red}{[gs: #1]}}

\input{macros.tex}

\AtBeginDocument{%
  }

\copyrightyear{2026}
\acmYear{2026}
\setcopyright{cc}
\setcctype{by}
\acmConference[WWW '26]{Proceedings of the ACM Web Conference 2026}{April 13--17, 2026}{Dubai, United Arab Emirates}
\acmBooktitle{Proceedings of the ACM Web Conference 2026 (WWW '26), April 13--17, 2026, Dubai, United Arab Emirates}
\acmPrice{}
\acmDOI{10.1145/3774904.3792678}
\acmISBN{979-8-4007-2307-0/2026/04}

\begin{document}

\title{Understanding Strategic Platform Entry and Seller Exploration: \\ A Stackelberg Model}

\author{Garrett Seo}
 \affiliation{%
 	\institution{Rutgers University}
 	\city{New Brunswick, NJ}
 	\country{USA}}
 \email{garrett.seo@rutgers.edu}

\author{Xintong Wang}
 \affiliation{%
 	\institution{Rutgers University}
 	\city{New Brunswick, NJ}
 	\country{USA}}
 \email{xintong.wang@rutgers.edu}

\author{David C. Parkes}
 \affiliation{%
 	\institution{Harvard University}
 	\city{Cambridge, MA}
 	\country{USA}}
 \email{parkes@eecs.harvard.edu}

\thanks{This research is funded in part by Rutgers SAS Research Grant in Academic
Themes. 
Our code is available at \url{https://github.com/chailab-rutgers/platform-entry}.
}

\begin{abstract}
Online market platforms play an increasingly powerful role in the economy. 
An empirical phenomenon is that platforms, such as Amazon, Apple, and DoorDash, also enter their own marketplaces, imitating successful products developed by third-party sellers. 
We formulate a Stackelberg model, where the platform acts as the leader by committing to an {\em  entry policy}: when will it enter and compete on a product? 
We study this model  through a theoretical and computational framework. 
We begin with a single seller,  and consider different
kinds of  policies for  entry.
We characterize the seller's optimal explore-exploit strategy via a Gittins-index policy, and give an algorithm to compute the platform's optimal entry policy.
%
We then consider multiple sellers, to account for competition and information spillover. Here, the Gittins-index characterization fails, and we employ deep reinforcement learning to examine seller equilibrium behavior.
Our findings highlight the incentives that drive platform entry and seller innovation, consistent with empirical evidence from markets such as Amazon and Google Play, with implications for regulatory efforts to preserve innovation and market diversity.

\end{abstract}



  


\maketitle

\section{Introduction}
\label{sec:intro}
Modern digital platforms like Amazon, Apple's App Store, and DoorDash have become dominant intermediaries for economic transactions, with their success built on third-party vendors and developers (``sellers'') who bring diverse services to the market.
These platforms also increasingly operate as direct competitors to sellers, by launching their own products. 
Examples include AmazonBasics, Apple's own apps in its App store, and the ``DoorDash Kitchens'' which emerged after the pandemic-driven food delivery boom.
The digital era has intensified this competition, bringing new  advantages to platforms. With unparalleled access to data, platforms can identify and imitate popular products at lower risk and cost than sellers~\citep{Zhu18}.
Furthermore, platforms have the power to promote their own items in search and recommendations, and operate at a scale unmatched by typically resource-constrained sellers.

The impact of platform entry on products is highly contested. 
It may benefit consumers (buyers) with lower prices and better quality, but it may also discourage third-party innovation if sellers expect the platform to imitate their products and take their profits.

This contentious dynamic has also become a focus of global antitrust enforcement.
In a landmark 2023 lawsuit, the U.S.~Federal Trade Commission, together with 17 states, sued Amazon, alleging that the company uses non-public seller data to illegally maintain its monopoly power by targeting lucrative markets for its own private-label products \citep{ftc_amazon_2023}. 
This was followed in March 2024 by a U.S.~Department of Justice lawsuit against Apple for allegedly suppressing innovative apps and technologies that could weaken the iPhone's dominance \citep{doj_apple_2024}. 
The European Union's {\em Digital Markets Act} (DMA), which became fully applicable in 2024, restricts platforms from using non-public business data to gain a competitive advantage~\citep{dma_2022}.\looseness=-1



This landscape raises critical questions for both platform operators and regulators: 
\textit{How should a platform decide whether and when to enter a product market, in balancing its own commercial interests against the health of its seller ecosystem? 
Do a platform's private objectives align with social welfare, or is regulatory intervention necessary?} 

To answer these questions, we develop a game-theoretic model of the strategic interaction between a platform and its sellers.
We model the platform-seller relationship as a Stackelberg game, where the platform, acting as the leader, first decides its policy as to when to enter and compete on each of one or more products. This positioning as the leader in the game reflects a platform's commitment power (e.g., through published fees or policies and exclusivity clauses). 
The sellers, after observing the platform's policy, then decide on their product innovation and sales strategy.

\subsubsection*{Our Contributions.}
We first study a single-seller model, where the seller's decision-making is modeled as a costly search problem.
By successfully adapting the Gittins index to this setting, we derive a closed-form expression for the seller's value of exploring an untested product (``exploration value'') which incorporates the platform's entry policy, allowing for a precise characterization of the seller's optimal explore/exploit policy. 
We show that the platform's revenue is piecewise monotonic with respect to the entry time or fee that parametrizes its policy, 
and develop an algorithm to optimize the platform's policy by identifying a finite set of Pareto optimal ones. 
We demonstrate this algorithm on three kinds of platform policies: a {\em global entry policy} that applies a universal entry time across products, a {\em global entry combined with transaction fees policy}, and a {\em heterogeneous entry policy} for product-specific entry times. 
We then generalize our model to a multi-seller environment to reflect the complex interactions on real-world platforms. 
As direct analytical solutions become intractable, we develop a multi-agent simulator and use deep reinforcement learning to find approximate Nash equilibria in the sellers' game. 
%
This allows us to solve for the platform's optimal entry policy under a range of 
distinct, strategically motivated environments (``clustered'' and ``diverse''), analyzing how platform entry affects different market scenarios.

The single and multi-seller simulations reveal several insights on the platform's strategic role. 
We find that a rational platform relies on seller exploration to discover high-demand products, which often increases overall consumer welfare (or social welfare)
and encourages product exploration when viable alternative products exist, relative to no platform entry. 
The platform's optimal entry is market-dependent: in ``clustered markets'' with a dominant product, a rational platform encourages seller diversification. 
The optimal timing, however, can be sensitive to the product's risk profile, e.g., high-stakes products require a longer protection period to incentivize seller innovation. 
In ``diverse markets'' with specialized sellers, an aggressive (early) entry policy can be destructive, forcing sellers to abandon their niches and cluster around selling safer products, reducing overall market diversity and welfare. 
This dynamic compels a rational platform to commit to delay its entry, demonstrating that while its private incentives are not perfectly aligned with social welfare, they also prevent very bad outcomes.


These findings suggest that, for a sufficiently forward-looking platform, the objective of revenue maximization is largely aligned with broader measures of market ecosystem health, such as product diversity and consumer welfare. 
The challenge, though, is to find the market-specific, non-myopic policies that a platform can use to achieve this alignment. 
This work provides a framework for understanding this strategic trade-off, offering a methodology for platform designers to optimize for long-term growth, and a lens for regulators to assess the market impact of platform competition.


\subsection{Related Work}
\label{sec:background}
\subsubsection*{Empirical Studies of Platform Entry.}
There are several empirical work documenting platform entry and competition with their third-party sellers. Amazon, rather than avoiding competition, tends to enter product spaces of high demand that are offered by many vendors \cite{Zhu18}.
Such entry decision increases the likelihood of third-party sellers exiting the platform. 
In the mobile app market, the threat of platform entry can cause developers to divert innovation efforts towards niche product categories to establish early competitive positions~\citep{Wen19}. 
Our work gives a game-theoretic model to provide counterfactual analysis of the strategic incentives behind these empirical observations.
%
\subsubsection*{Economic Models of Hybrid Platforms.} 
Prior literature has developed economic models of hybrid platforms to analyze the trade-offs of their dual role~\citep{Anderson2022}, the impact of data regulations~\citep{Madsen2023}, and specific strategies like self-preferencing in search rankings \citep{Hagiu2022} and the use of vertical control mechanisms \citep{Kang2022EC}.
We contribute to this line of work by offering a computational model focused specifically on how a platform's entry policy affects seller innovation, product diversity, and market efficiency.

\subsubsection*{Contract Design for Exploration.} 
The seller's problem in our model can be viewed as a variant of costly sequential search, rooted in the classic ``Pandora's Box''  problem~\citep{Weitzman79} and multi-armed bandit frameworks~\citep{Gittins1979}. 
Related work studies how a principal can guide an agent's search using explicit contracts, such as linear payment schemes \citep{dutting2023} or formal delegation mechanisms \citep{Bechtel2022, hoefer2024contract,ivanov2024principalagentreinforcementlearningorchestrating}. 
In our model, the platform's entry policy provides an implicit contract, shaping the exploration behavior of strategic sellers.

\subsubsection*{RL for Economic Platform Design.} 
Our approach aligns with recent work using RL and simulation to design and understand complex economic systems, including dynamically setting reserve prices in auctions~\citep{ShenPLZQHGDLT20}, selling user impressions to advertisers~\citep{Tang17abc}, designing tax policies~\citep{Zheng2021},  optimizing user satisfaction for recommender systems~\citep{Chen2019,Zhan2021}, and designing sequential price mechanisms~\citep{BreroEGPR21}.
Consistent with some of this literature, we use a Stackelberg game to model the platform-user interaction, a framework previously applied to analyze strategic problems like collusion mitigation~\citep{Gianluca22} and platform fee-setting under market shocks~\citep{platform_shock23}.

\subsubsection*{Gittins Index.}
The multi-armed bandit (MAB) problem models a sequential decision problem where an agent balances exploring less well-understood options with exploiting better known ones to maximize cumulative reward. 
For a specific class of MAB problems, the Gittins index provides a provably optimal solution~\citep{gittins1974dynamic}. This powerful result allows a complex, multi-dimensional optimization problem to be decomposed into a series of independent calculations, one for each arm. At each decision point, the optimal policy is to simply choose the arm with the highest current Gittins index. We demonstrate the applicability of the Gittins index to the seller problem in the single-seller setting. 
We defer details on the definition and assumptions to Appendix~\ref{app:gindex}.

\section{A Single-Seller Model with Platform Entry}
%

\label{sec:single_seller_model}
We first consider the interaction between a platform and a single third-party seller, 
and analyze how the platform's policy influences the seller's exploratory behavior. 
%

\subsection{A Stackelberg Model of Platform Entry}
\subsubsection{Shared Product Space.}
We consider a set of $M$ potential products, each initially in an undeveloped (\sunexplore) state. 
The demand or reward for each product is unknown on the platform before being sold. 
We assume the platform only enters products after a seller has explored a product and revealed its demand (e.g., from public sales ranking of products).

\subsubsection{The Seller's Problem.}
To explore an undeveloped product $j$, the seller incurs a one-time innovation cost $c_j$. 
Upon exploration, the product's demand is revealed: it transitions to a ``good'' state (\sgood) with reward $r^g_j$ and probability $p_j$ or a ``bad'' state (\sbad) with reward $r^b_j$ and probability $1-p_j$. 
We assume $r^g_j > r^b_j \ge 0$ and that the seller has accurate priors of $p_j$, $r^g_j$, and $r^b_j$, informed by sources such as market research, historical data, or similar products. 
These parameters capture essential differences in product types, reflecting variations in development costs and market positioning (e.g., luxury versus practical, niche versus mass market).
The innovation cost can also reflect a seller's expertise within a product domain. 
%

The seller has a capacity-constraint of offering one product at a time, and seeks to maximize their total expected discounted reward, given the platform's policy. 

\subsubsection{The Platform's Policy.}
The platform commits to a policy $\boldsymbol{\pi_p}$ that includes an \textit{entry-time parameter} $T_{p_j}$ as part of its strategic design. 
This parameter specifies the number of timesteps the platform waits before entering a developed product $j$ that has been revealed to be in its ``good'' state by the seller, transitioning product $j$ to an entered (\senter) state.
This reflects data patents, or empirical observations that platforms such as Amazon typically enter only after products demonstrate strong sales and positive reviews~\citep{Zhu18}. 
For simplicity, we assume that once the platform enters, it captures the entire reward stream $r^g_j$.\looseness=-1

The platform can further introduce a fraction parameter $\alpha$, representing the reward split between the platform and the seller. 
The parameter $\alpha$ can be interpreted as a \textit{transaction fee} that applies regardless of the product’s realized state. 

Given these components, in Section~\ref{sec:seller_policy}, we analyze three platform policy settings: a {\em global entry $T_p$ for all products}, a {\em global entry $T_p$ with transaction fee $\alpha$}, and a {\em heterogeneous entry $\mathbf{T_p} = (T_{p_1}, T_{p_2}, \dots, T_{p_M})$ allowing distinct entry times for each product}.\looseness=-1

\subsection{Seller's Optimal Policy Under $\boldsymbol{\pi_p}$}
Without platform entry, the seller faces a multi-arm bandit problem where the product opportunities are independent stochastic processes.
In this setting, the Gittins-index policy provides a provably optimal strategy~\citep{Gittins1979}.

We now introduce platform policy and denote the resulting problem by $\mathcal{M}$. The platform policy modifies the state-transition dynamics, so we model each product as a Markov chain with state space $S_j =\{\sunexplore, \sgood, \sbad, \senter\}$, representing the \textit{undeveloped, good, bad, and entered by platform} states, respectively. 
We  derive a closed-form, piece-wise expression for the Gittins index that incorporates the seller's optimal stopping decision in response to a platform entry $T_{p_j}$.
For each product $j$, we calculate the Gittins index by identifying the optimal stopping time. 
We consider the optimal stopping rule by partitioning the state space into a \textit{stopping set} $\mathcal{S}$ and a \textit{continuation set} $\mathcal{C}=S_j \setminus \mathcal{S}$. 
If the product is in a state $s \in \mathcal{C}$, the seller continues to sell $j$; if $s \in \mathcal{S}$, the seller stops selling $j$
In our platform entry model, it suffices to consider the following three different stopping rules, with their corresponding indices.

\begin{enumerate}[leftmargin=*]
    \item Stopping rule $\tau_1$: $\mathcal{S} = \{\sunexplore\}$. The seller does not explore.
    \item Stopping rule $\tau_2$: $\mathcal{S} = \{\senter\}$. The seller should continue to gain rewards from a bad state, and stop at platform entry.
    \item Stopping rule $\tau_3$: $\mathcal{S} = \{\sbad, \senter\}$. The seller stops if the product is revealed to be bad or if the platform enters. 
\end{enumerate}

We do not consider state $S_j = G$ to be in the stopping set as $r_j^g$ is the highest reward achievable in any given state. 
We defer the detailed closed-form derivation of these stopping-time indices under platform policy, $\boldsymbol{\pi_p}$ to Appendix~\ref{app:gindex_derivation_single_platform}.
%
Let $G_j^{(k)}(\sunexplore; \boldsymbol{\pi_p})$ denote the Gittins index associated with a specific stopping rule $\tau_k$ for the undeveloped product $j$. 
The Gittins index for product $j$ is the maximum value across all stopping rules:
\[
    G_j(\sunexplore; \boldsymbol{\pi_p}) = \max \left\{G_j^{(1)}(\sunexplore; \boldsymbol{\pi_p}), G_j^{(2)}(\sunexplore; \boldsymbol{\pi_p}), G_j^{(3)}(\sunexplore; \boldsymbol{\pi_p}) \right\}
\]

However, the optimality of the Gittins-index policy relies on independence across products, so we first check whether the independence among products is preserved, and under what conditions.


\begin{proposition}
    A platform's policy $\boldsymbol{\pi_p}$ preserves the products as independent stochastic processes, only if the seller acts optimally in response. 
    The seller's optimal policy is the Gittins-index policy. 
\end{proposition}

\begin{proof}
It suffices to prove the claim for heterogeneous entry times 
$\mathbf{T}_p$. A global policy $T_p$ corresponds to the special case 
$T_{p_j}=T_p$ for all products $j$, while a transaction fee $\alpha$ 
uniformly rescales rewards by $(1-\alpha)$. Hence the results for 
$T_p$ and $(T_p,\alpha)$ follow directly.
The proof proceeds in four steps:

(1) Construct a rested bandit $\mathcal{M}'$ from the current version $\mathcal{M}$.

(2) Show $V_\pi^{\mathcal{M}'} \ge V_\pi^{\mathcal{M}}$ for any seller policy $\pi$.

(3) Characterize a set of seller policies ${\pi'}$ where $V_\pi^{\mathcal{M}'} = V_\pi^{\mathcal{M}}$.

(4) Show the Gittins-index policy belongs to this set.

\paragraph{Step 1 (Rested Conversion $\mathcal{M'}$).} The seller's problem $\mathcal{M}$ is described by a market state $\mathbf{x}(t) = (\mathbf{S}, \mathbf{C})$. $\mathbf{S}$ denotes the product state vector, where each $S_j$ takes values in the state space defined previously. $\mathbf{C}$ denotes the product counter vector where each $C_j \in \{\emptyset, T_{p_j}, T_{p_j} -1, \dots, 2, 1, 0\}$ tracks the time left before the platform enters product $j$. The counter is inactive ($C_j =\emptyset$) when product $j$ is in state $\sunexplore$ or $\sbad$, and equals $0$ when the product has entered state $\senter$. 

$\mathcal{M}$ is not made up of independent stochastic products. In particular, suppose product $j$ is in state
\[
x_j(t) = (\sgood, C_j),
\]
meaning it is in its good state and the platform will enter after $C_j$ timesteps.
If the seller instead sells a different product $j' \neq j$ at time $t$, then the counter for product $j$ still decreases, so its state updates to
\[
x_j(t+1) =
\begin{cases}
(\sgood, C_j-1), & \text{if } C_j>1,\\
(\senter, 0), & \text{if } C_j=1.
\end{cases}
\]
Thus, the state of product $j$ changes despite selling another product $j'$.

To recover independence, we define a \emph{rested} version of the problem, denoted $\mathcal{M}'$. $\mathcal{M}'$ matches $\mathcal{M}$ in all state transitions, except for state $\sgood$. 
Specifically, in $\mathcal{M}'$, if product $j$ is in state $x_j(t)=(\sgood, C_j)$ and the seller chooses some other product $j'\neq j$, then product $j$ remains unchanged:
\[
x_j(t+1)= (\sgood, C_j).
\]
However, if the seller chooses product $j$ at time $t$, then we keep the original transition:
\[
x_j(t+1) =
\begin{cases}
(\sgood, C_j-1), & \text{if } C_j>1,\\
(\senter, 0), & \text{if } C_j=1.
\end{cases}
\]

\paragraph{Step 2 (Dominance of $\mathcal{M'}$).} We now show that the seller always makes more utility in $\mathcal{M}'$. We define a trajectory
\[\tau = (\mathbf{x}(0), a(0), \mathbf{x}(1), a(1), \dots)\]
which is the sequence of market states and actions generated in $\mathcal{M}'$ or $\mathcal{M}$. For a seller policy $\pi$, we define the expected discounted utility from initial state $\mathbf{x}(0)$ as
\[
V_\pi\big(\mathbf{x}(0)\big)
= \mathbb{E}_{\tau \sim p_\pi\big(\cdot \mid \mathbf{x}(0)\big)}
\left[\sum_{t=0}^{\infty} \gamma^t \, R\big(\mathbf{x}(t), \pi(\mathbf{x}(t)\big)\right]
.
\]
Here, $\gamma \in (0,1)$ is the discount factor and $R\big(\mathbf{x}(t),a(t)\big)$ is the reward at time $t$. 

Further, the trajectory distribution $p_\pi$ is defined as: 
\[
p_\pi(\tau \mid \mathbf{x}(0))
= \prod_{t=0}^{\infty} \pi\big(a(t) \mid \mathbf{x}(t)\big)\, \mathcal{T}\big(\mathbf{x}(t+1) \mid \mathbf{x}(t), a(t)\big),
\]
where $\mathcal{T}$ denotes the transition probability function, and $p_\pi$ describes the probability of the seller playing some trajectory $\tau$ under policy $\pi$.  

To show that the expected reward of the seller in $\mathcal{M'}$ is greater than or equal to in $\mathcal{M}$, or that $V_\pi^\mathcal{M'}\big(\mathbf{x}(0)\big) \geq V_\pi^\mathcal{M}\big(\mathbf{x}(0)\big)$, we show that the reward from any given trajectory in $\mathcal{M'}$ is greater than that in $\mathcal{M}$. 

Fix an arbitrary trajectory $\tau$, and consider any time $t$ along this trajectory with state-action pair $\big(\mathbf{x}^\mathcal{M'}(t),a(t)\big)$ and $\big(\mathbf{x}^\mathcal{M}(t),a(t)\big)$ for $\mathcal{M'}$ and $\mathcal{M}$ respectively.
We claim that, for every action $a(t)$ in the sequence,
\[
R\big(\mathbf{x}^\mathcal{M'}(t),a(t)\big) \;\ge\; R\big(\mathbf{x}^{\mathcal{M}}(t),a(t)\big).
\]

\noindent We compare rewards across actions.

\noindent\textbf{Action $a(t)=0$.}
No product is sold, and the reward is $0$ in both problems:
\[
R\big(\mathbf{x}^{\mathcal{M}'}(t),a(t)\big)
=
R\big(\mathbf{x}^{\mathcal{M}}(t),a(t)\big)
=0.
\]

\noindent\textbf{Action $a(t)=j$.}
The seller sells product $j$ at time $t$. We distinguish three product states. 

\noindent\emph{Undeveloped.}
If $x^{\mathcal{M}'}_{j}(t) = (\sunexplore,\emptyset)$, then necessarily
$x^{\mathcal{M}}_{j}(t) = (\sunexplore,\emptyset)$ as well. Therefore,
\[
R\big(\mathbf{x}^{\mathcal{M}'}(t),a(t)\big)
=
R\big(\mathbf{x}^{\mathcal{M}}(t),a(t)\big)
= -c_j + 
\begin{cases}
r^g_j, & \text{w.p. } p_j,\\
r^b_j, & \text{w.p.  } 1-p_j.
\end{cases}
\]

\noindent\emph{Bad.}
If $x^{\mathcal{M}'}_{j}(t) = (\sbad,\emptyset)$, then
$x^{\mathcal{M}}_{j}(t) = (\sbad,\emptyset)$ as well, and hence
\[
R\big(\mathbf{x}^{\mathcal{M}'}(t),a(t)\big)
=
R\big(\mathbf{x}^{\mathcal{M}}(t),a(t)\big)
=
r^b_j.
\]

\noindent\emph{Good.}
Suppose $x^{\mathcal{M}'}_{j}(t) = (\sgood, C^{\mathcal{M}'}_{j})$.
Counters evolve differently across the two problems: in $\mathcal{M}'$
the counter decreases only when product $j$ is selected, whereas in
$\mathcal{M}$ it decreases whenever the product is in the good state.
Hence $C^{\mathcal{M}'}_{j} \ge C^{\mathcal{M}}_{j}$.

If $C^{\mathcal{M}}_{j}=0$, then
$x^{\mathcal{M}}_{j}(t)=(\senter,0)$ while
$x^{\mathcal{M}'}_{j}(t)=(\sgood,C^{\mathcal{M}'}_{j})$,
so
\[
R(\mathbf{x}^{\mathcal{M}'}(t),a(t))
= r^g_j
>
R(\mathbf{x}^{\mathcal{M}}(t),a(t))
=0 .
\]
Otherwise the product remains in state $\sgood$ in both $\mathcal{M}'$ and $\mathcal{M}$, and the rewards coincide,
yielding $R^{\mathcal M'} = R^{\mathcal M} = r^g_j$.

Across all actions, the reward gained by the seller in $\mathcal{M}'$ is greater than or equal to the reward in $\mathcal{M}$ at every time step along a trajectory $\tau$.
Because the two problems induce the same trajectory distribution under any fixed seller policy $\pi$ (i.e., $p^{\mathcal{M}'}_\pi(\tau \mid \mathbf{x}_0)=p^{\mathcal{M}}_\pi(\tau \mid \mathbf{x}_0)$), it follows that the expected discounted utility in $\mathcal{M}'$ is weakly higher, and $ V_\pi^{\mathcal{M}'}\big(\mathbf{x}(0)\big) \geq V_\pi^{\mathcal{M}}\big(\mathbf{x}(0)\big)$. 

\paragraph{Step 3 (Policies Achieving $V_\pi^{\mathcal{M}'} = V_\pi^{\mathcal{M}}$).} We now describe the set of policies $\{\pi'\}$ for which $V_\pi^{\mathcal{M}'}\big(\mathbf{x}(0)\big) = V_\pi^{\mathcal{M}}\big(\mathbf{x}(0)\big)$. Previously, we showed that $\mathcal{M'}$ can yield a strictly greater reward than $\mathcal{M}$ only when some product $j$ is in the good state $x^{\mathcal{M}'}_j(t) = (G, C_j^{\mathcal{M'}})$, since the counter satisfies $C_j^\mathcal{M'} \geq C_j^\mathcal{M}$. In order for both counters to be equivalent, we must have $\pi'$ satisfy that if $x_j(t) = (\sgood, C_j)$, then $\pi'\big(\mathbf{x}(t)\big) = j$. In this way, the counter decreases the same way in both problems, and we can interpret $\pi'$ as a seller policy that always continues to sell a product in its $\sgood$ state until entry.

Since $\mathcal{M}'$ is a MAB problem with independent stochastic arms, the Gittins-index policy $\pi^*$ is optimal ~\citep{weber-proof}.

\paragraph{Step 4 (Gittins Optimality).} To complete the proof, we must show that the Gittins-index policy $\pi^*$ belongs to the set of policies $\{\pi'\}$. In order for $\pi^* \in \{\pi'\}$, $\pi^*$ satisfies that if $x_j(t) = (\sgood, C_j)$, then $\pi^*\big(\mathbf{x}(t)\big) = j$. 

Suppose the seller follows $\pi^*$ and at time $t$ product $j$ is in the good state,
$x_j(t)=(\sgood,C_j)$. Since all products are initially undeveloped in $\mathbf{x}(0)$,
there exists some $t'<t$ at which product $j$ was first selected, with
$x_j(t')=(\sunexplore,\emptyset)$. By the Gittins-index policy, $j$ maximized the index at
time $t'$, i.e.,
\[
G_j(\sunexplore;T_{p_j}) \ge G_{j'}
\quad \forall\, j'\neq j.
\]
Moreover, since
\[
G_j(\sgood;T_{p_j}) > G_j(\sunexplore;T_{p_j}),
\]
product $j$ continues to maximize the index at time $t$, implying
\[
\pi^*(\mathbf{x}(t)) = j.
\]
\end{proof}

The proof above highlights how platform policies shape optimal exploration incentives. In some market environments, however, exploration is effectively predetermined. We define a \emph{dominating product} $j^\text{dom}$ as a product whose undeveloped Gittins index exceeds that of any other product under any platform policy:
\[
G_{j^{\mathrm{dom}}}(U;\boldsymbol{\pi}_p)
>
G_j(U;\boldsymbol{\pi}_p)
\quad
\forall j \neq j^{\mathrm{dom}},\;
\forall \boldsymbol{\pi}_p .
\]

In experiments, we focus on market environments without dominating products, since sellers would always explore $j^\text{dom}$ first, limiting the influence of platform policies on exploration. 

\subsection{Finding the Optimal Platform Policy}
\label{sec:seller_policy}
The platform's objective is to choose the optimal $\boldsymbol{\pi^*_p}$, that maximizes its own expected discounted utility, anticipating the seller's optimal response. 
Instead of exhaustively searching over an infinite set of platform policies $\boldsymbol{\pi_p}$, we show that the policy space can be partitioned into regions, each with a distinct optimal seller strategy. 
We further find that the platform's optimal policy lies among a finite set of candidate points across these regions.

We first show that the platform's utility, $u_p$, is a piecewise function of its policy $\boldsymbol{\pi_p}$, with branches defined by the optimal seller strategy $\pi^*_s$ that follows from $\boldsymbol{\pi_p}$. 
For this section, we simplify the market state as $\mathbf{x} = (\mathbf{S}, t)$, where $\mathbf{S}$ is the product state vector with each $S_j \in \{\sunexplore, \sbad, \sgood, \senter\}$ indicating the state of product $j$, and $t$ denotes the timestep. 
The utility for the platform can be described recursively:

\[
\begin{aligned}
U_p(\mathbf{S},t;\boldsymbol{\pi}_P) =
\begin{cases}
0, & j^* = \varnothing, \\[4pt]
F_{\text{bad-cont}}(r^b_{j^*},t,\boldsymbol{\pi}_P),
  & S_{j^*} = \sbad, \\[4pt]
\begin{aligned}
 &p_{j^*}^b\!\Big(F_{\text{bad}}(r^b_{j^*},t,\boldsymbol{\pi}_P) \\
 &\qquad + U_p(\mathbf{S}^{(S_{j^*}\leftarrow \sbad)},t+1;\boldsymbol{\pi}_P)\Big) \\
 &+ p_{j^*}^g\!\Big(F_{\text{good}}(r^g_{j^*},t,\boldsymbol{\pi}_P) \\
 &\qquad + U_p(\mathbf{S}^{(S_{j^*}\leftarrow \senter)},t+T_{p_j};\boldsymbol{\pi}_P)\Big)
\end{aligned}
 & S_{j^*} = \sunexplore
\end{cases}
\end{aligned}
\]
where $j^* = \pi^*_s(\mathbf{S},t, \boldsymbol{\pi}_P)
     = \arg\max_j G_j(S_j, t;\boldsymbol{\pi}_P)$, or equivalently, the best product to sell at state $(\mathbf{S}, t)$ by the Gittins index.

$F_{\text{bad-cont}}$, $F_{\text{bad}}$, and $F_{\text{good}}$ are closed-form expressions of the platform's reward that depend on the policy setting. 
$F_{\text{bad-cont}}$ represents the ongoing reward that the platform obtains when the seller continues to sell a product that is realized in its bad state for the remainder of the game.
$F_{\text{bad}}$ and $F_{\text{good}}$ correspond to the reward the platform receives when the seller's product transitions to its bad or good state, respectively. 


We identify where the seller’s optimal policy may change with $\boldsymbol{\pi_p}$ by defining a boundary set $\mathcal{B}$, made up of three types:

\begin{enumerate}[leftmargin=*]
    \item \textit{Zero boundary}: $b_j^0 = G_j(\sunexplore; \boldsymbol{\pi_p}) = 0$. 
    The seller is indifferent between exploring and not exploring product $j$. 
    \item \textit{Bad-indifference boundary}: $b_{j, j'}^{B} = G_j(\sbad; \boldsymbol{\pi_p}) - G_{j'}(\sunexplore; \boldsymbol{\pi_p}) = 0$. The seller is indifferent between selling product $j$, which has been realized in its bad state, and exploring an undeveloped product $j'$. 
    \item \textit{Undeveloped-indifference boundary}: $b_{j, j'}^\sunexplore = G_j(U; \boldsymbol{\pi_p}) - G_{j'}(U; \boldsymbol{\pi_p})= 0$. The seller is indifferent between exploring product $j$ and $j'$.
\end{enumerate}
Each boundary $b \in \mathcal{B}$ partitions the policy space $\boldsymbol{\pi_p}$ into the sets $\{\boldsymbol{\pi_p}: b(\boldsymbol{\pi_p}) > 0\}$ and $\{\boldsymbol{\pi_p}: b(\boldsymbol{\pi_p}) < 0\}$. 
We defined a \textit{region} as the non-empty intersection of such sets across all boundaries $b \in \mathcal{B}$, representing a subset of the platform policy space where the seller's optimal choice of $j^*$ remains identical for a given market state $(\mathbf{S}, t)$. Thus, each piece of $u_p(\boldsymbol{\pi_p})$ corresponds to a region with a fixed seller strategy.

Let $\mathcal{R}(\mathcal{B}) = \{R_1, R_2, ..., R_k\}$ be the collection of all nonempty regions induced by $\mathcal{B}$, where the seller strategy remains fixed for all $\boldsymbol{\pi_p} \in R_i$. We can now define $u_p$ over all branches $R_i \in \mathcal{R}(\mathcal{B})$ as: 
\[
u_p(\boldsymbol{\pi_p}) =
\begin{cases}
u_p(\boldsymbol{\pi_p} | R_1) = U_p(\mathbf{S}^U, 0; R_1), & \text{if } \boldsymbol{\pi_p} \in R_1, \\[0.6em]
\;\;\vdots & \;\;\vdots \\[0.6em]
u_p(\boldsymbol{\pi_p} | R_k) = U_p(\mathbf{S}^U, 0; R_k), & \text{if } \boldsymbol{\pi_p} \in R_k,
\end{cases}
\quad R_i \in \mathcal{R}(\mathcal{B}).
\]

To maximize $u_p$, we maximize over all regions. 
Naively, this can be done by iterating through all possible policies within a region. 
Instead, we show that for every region $R_i \in \mathcal{R}(\mathcal{B})$, $u_p(\boldsymbol{\pi_p} |R_i)$ is monotone with respect to each component of $\boldsymbol{\pi_p}$. 
Thus, maximizing a region reduces to maximizing over a set of \textit{Pareto optimal} platform strategies within the region for which no other strategy improves all monotone components simultaneously. 

Given a policy setting, every monotonically increasing component (e.g., transaction fee) is bounded above and every monotonically decreasing component (e.g., entry time) is bounded below within $R_i$, so the set of Pareto optimal strategies is finite.
Below we analyze three platform policy settings.

\subsubsection{Global Entry}
A global entry policy is defined with a single platform entry time $\boldsymbol{\pi_p} = T_p \in \mathbb{N}$, the same for all products. 
Denote $\gamma$ as the platform discount factor. 
The closed-form functions $F$, representing the platform's reward, are: 
\[
\begin{aligned}
F_{\text{bad-cont}}(r^b_{j},t,\boldsymbol{\pi}_P=T_p) &=
F_{\text{bad}}(r^b_{j},t,\boldsymbol{\pi}_P=T_p) = 0, \\
F_{\text{good}}(r^g_{j},t,\boldsymbol{\pi}_P=T_p) &= r_{j}^g \frac{\gamma^{t+T_p}}{1-\gamma} . 
\end{aligned}
\]

$F_{\text{good}}$ increases as $T_p$ decreases, 
\[F_{\text{good}}(r^g_{j}, t, T_{p} + 1) < F_{\text{good}}(r^g_{j}, t, T_{p}), \; \forall \; T_p \in \mathbb{N}.\]
Since $u_p(T_p)$ is a weighted sum of $F_{\text{good}}$, $u_p(T_p)$ decreases monotonically with $T_p$. Further, $T_p$ is bounded below, so every region $R_i$ can be maximized over a finite set of Pareto optimal points. Because each region $R_i$ is one-dimensional in $T_p$, identifying this set reduces to selecting the smallest $T_p$ in that region.

\subsubsection{Global Entry and Transaction Fee}
We extend the platform's policy $\boldsymbol{\pi_p} = (T_p, \alpha)$ to include a {\em transaction fee, $\alpha \in [0,1]$}. 
In this case, the closed-form functions $F$ are: 
\[
\begin{aligned}
F_{\text{bad-cont}}(r^b_{j},t,\boldsymbol{\pi}_P=(T_p, \alpha)) &= \alpha r_{j}^b \frac{\gamma^t}{1-\gamma},  \\[4pt]
F_{\text{bad}}(r^b_{j},t,\boldsymbol{\pi}_P=(T_p, \alpha))  &= \alpha r_{j}^b \gamma^t,  \\[4pt]
F_{\text{good}}(r^g_{j},t,\boldsymbol{\pi}_P=(T_p, \alpha)) &= \alpha r_{j}^g \frac{\gamma^t (1-\gamma^{T_p})}{1-\gamma} + r_{j}^g \frac{\gamma^{t+T_p}}{1-\gamma} . 
\end{aligned}
\]
All $F$ are linear and nonnegative in $\alpha$, and thus increase monotonically with $\alpha$. 
Meanwhile, $F_{\text{good}}$ remains decreasing in $T_p$:
\[
\begin{aligned}
F_{\text{good}}(T_p+1) - F_{\text{good}}(T_p) 
&= - r_{j}^g \frac{\gamma^{t+T_p} (1-\alpha)}{1-\gamma} \le 0.
\end{aligned}
\]
Since $\gamma \in (0,1)$, $r_{j}^g \ge 0$, and $\alpha \in [0,1]$, the difference is always nonpositive. 
Thus, $F_{\text{good}}(T_p+1) \le F_{\text{good}}(T_p)$ for all $T_p \in \mathbb{N}$, and $F_{\text{good}}$ decreases monotonically with $T_p$.

$u_p$ is a weighted sum of the $F$ functions, and therefore, increases monotonically with $\alpha$ and decreases monotonically with $T_p$ within each region $R_i$. As established, $T_p$ is bounded below, and $\alpha$ is also bounded above by $1$, so every region $R_i$ can be maximized over a finite set of Pareto optimal points. Formally, a policy $(T_p, \alpha)$ is Pareto optimal if there exists no other point $(T_p', \alpha')$ such that $T_p' \le T_p$ and $\alpha' \ge \alpha$ with at least one inequality strict.

\subsubsection{Heterogeneous Entry}
We consider a flexible entry policy $\boldsymbol{\pi_p} = \mathbf{T_p} = (T_{p_1}, T_{p_2}, \dots, T_{p_M}) \in \mathbb{N}^M$, where the platform can set {\em an independent entry policy $T_{p_j}$ for each product $j$}.  The closed-form functions F are: 
\[
\begin{aligned}
F_{\text{bad-cont}}(r^b_{j},t,\boldsymbol{\pi}_P=\mathbf{T_p}) &=
F_{\text{bad}}(r^b_{j},t,\boldsymbol{\pi}_P=\mathbf{T_p}) = 0, \\[4pt]
F_{\text{good}}(r^g_{j},t,\boldsymbol{\pi}_P=\mathbf{T_p}) &= r_{j}^g \frac{\gamma^{t+T_{p_j}}}{1-\gamma} . 
\end{aligned}
\]

Here, $F_{\text{good}}$ monotonically decreases with each $T_{p_j}$ and is bounded below, as shown in the global entry case. 
Thus, every region $R_i$ can be maximized over a finite set of Pareto optimal points. Formally, a policy $\mathbf{T_p} =(T_{p_1}, T_{p_2}, \dots, T_{p_M})$ is Pareto optimal if there exists no other $\mathbf{T'_p} = (T'_{p_1}, T'_{p_2}, \dots, T'_{p_M})$ such that $T_{p_j}' \le T_{p_j}$ for all $j$ with at least one inequality strict.


\subsubsection{Computation and Complexity} We summarize the algorithm for computing the optimal platform policy in Appendix \ref{app:platpolalg}. We provide intuition for the global entry and heterogeneous entry policy settings with toy examples in Appendix \ref{app:globaltoyex} and \ref{app:heterotoyex}. 

The number of regions in the policy space depends on the setting. For global entry and global entry with transaction fees, the number of regions grows polynomially in the number of products due to the seller indifference boundaries. 
However, heterogeneous entry creates an $M$-dimensional policy space, causing the number of regions to grow exponentially with the number of products. Monotonicity and boundedness allow us to focus only on Pareto-optimal strategies within each region, despite this exponential complexity.




\subsection{An Illustration of Platform Policy Choice}
In this section, we construct simple, representative market environments to study how a strategic platform may tailor its entry and fee policies according to market structure. 
We focus on identifying conditions under which platform incentives align with or diverge from seller exploration, and explore how regulatory interventions can mitigate platform's excessive profit extraction.

\begin{table}[h]
\centering
\caption{Type A represents a product with moderate, stable payoffs and a relatively low cost, whereas Type B a riskier product with the potential of higher reward but higher cost.}
\label{tab:ProductTypes}
\begin{tabular}{l@{\hspace{2.5em}}l@{\hspace{2.5em}}l@{\hspace{2.5em}}}
\hline\hline
 & \text{Type A} & \text{Type B} \\
\hline
Cost & 50 & 120 \\
Reward & 100, 50 & 200, 0 \\
Probability & 0.5, 0.5 & 0.2, 0.8 \\
\hline\hline
\end{tabular}
\end{table}

We use two types of products in Table~\ref{tab:ProductTypes} to construct environments reflecting different distributions of product opportunities:
\begin{itemize}[leftmargin=*]
    \item A market with three Type A and one Type B products (3A1B), 

    \item A market with one Type A and three Type B products (1A3B). 
    
    
\end{itemize}

Table~\ref{tab:utility_metrics_combined} presents the optimal platform policies and the associated utility and exploration metrics in the two markets. 
We set the seller’s discount factor to $\gamma_s = 0.9$, and assume a more forward-looking platform with $\gamma_p = 0.95$.
We highlight the following:



\begin{enumerate}[leftmargin=*,itemsep=0.2cm]
    \item \textit{Rational} platform entry encourages exploration. \\
    This is reflected in the consistent increase in the number of products explored relative to the no-entry case. 
    A rational platform avoids premature entry, since its profits remain partially aligned with seller exploration. 
    \item \text{Market composition shapes optimal platform behavior.}\\
    In markets dominated by safe products with predictable demand and low innovation costs (e.g., 3A1B), the platform optimally sets higher fees (40\%) to reliably capture a steady revenue stream. This scenario likely reflects many real-world markets, driven by stable demand and incremental innovation.
    
    By contrast, in markets where demand is less predictable and success depends on innovating riskier products (e.g., 1A3B), the platform benefits from setting lower transaction fees (8\%), better aligning its incentives with seller exploration and enabling it to imitate and monetize more new products.
    %
    \item High transaction fees reduce seller utility and exploration.\\
    In markets like 3A1B, the platform prioritizes profit extraction via higher fees over earlier entry to give sellers time to recoup from costs and explore later. 
    This suppresses early exploration and buyer utility. 
    Imposing caps on transaction fees limits the platform from extracting excessive profits while encouraging earlier entry, boosting product exploration and buyer utility.
    %
    \item Heterogeneous entry improves flexibility and buyer utility, but may hurt sellers.\\
    Allowing product-specific entry times enables the platform to imitate low-cost products earlier and high-cost products later, at the expense of seller profits.
    By placing a minimum entry barrier, we limit the platform from imitating too early, maintaining higher buyer utility while improving seller profit.   
\end{enumerate}

\begin{table*}[t]
\caption{Agent utility and seller exploration metrics on markets with different product compositions.}
\label{tab:utility_metrics_combined}
\centering
\footnotesize 
\begin{subtable}[t]{0.49\textwidth}
\centering
\caption{Environment 3A1B}
\label{tab:3Safe}
\begin{tabularx}{\linewidth}{>{\raggedright\arraybackslash}X c c c c}
\toprule
\textbf{Policy Setting} & \textbf{Platform} & \textbf{Seller} & \textbf{Buyer} & \textbf{Prod. Explored} \\
\midrule
No platform entry or fee & 0 & 865 & 2174 & 2.38\\
$T_p^* = 3$ & 2332 & 514 & 3447 & 3\\
$(T_p^*, \alpha^*) = (4,0.4)$ & 2633 & 292 & 3285 & 3\\
$\mathbf{T_p^*} = (3,3,3,8)$ & 2724 & 525 & 3967 & 4\\
\rowcolor{gray!15} Fee cap $\alpha \le 0.2$ $(T_p^*, \alpha^*) = (3,0.2)$ & 2555 & 386 & 3447 & 3\\
\bottomrule
\end{tabularx}
\end{subtable}
\hfill
\begin{subtable}[t]{0.49\textwidth}
\centering
\caption{Environment 1A3B}
\label{tab:3Risky}
\begin{tabularx}{\linewidth}{>{\raggedright\arraybackslash}X c c c c}
\toprule
\textbf{Policy Setting} & \textbf{Platform} & \textbf{Seller} & \textbf{Buyer} & \textbf{Prod. Explored} \\
\midrule
No platform entry or fee & 0 & 876 & 2510 & 2.9\\
$T_p^* = 8$ & 1814 & 551 & 3098 & 4\\
$(T_p^*, \alpha^*) = (8,0.08)$ & 1921 & 485 & 3098 & 4\\
$\mathbf{T_p^*} = (1,7,7,7)$ & 2205 & 354 & 3259 & 4\\
\rowcolor{gray!15} Entry cap $T_{p_j} \ge 5$ $\mathbf{T_p^*} = (5,8,8,8)$ & 2004 & 492 & 3217 & 4\\
\bottomrule
\end{tabularx}
\end{subtable}
\caption*{
Seller exploration is evaluated by expected products explored. 
Total buyer utility is the sum of discounted rewards from offered products, reflecting realized demand. 
Shaded rows indicate settings with a cap on transaction fees (left) or entry barrier (right).}
\end{table*}

\section{A Multi-Seller Model with Platform Entry}
\label{sec:multi_seller_model}


We now consider a platform with multiple sellers, capturing richer strategic dynamics, such as information spillover and market congestion.
Here, the assumptions required for the optimality of the Gittins index for a seller's problem no longer hold.
Exploration by a seller now generates a public signal about the associated product, affecting the beliefs and strategies of all others. 
At the same time, as multiple sellers crowd into the same product space, individual payoffs may decline. 
Thus, the problem is no longer a set of independent search problems but a complex, non-stationary game.
To study this more realisitic scenario, we extend our model and use multi-agent reinforcement learning to identify and analyze the resulting equilibria. 

We focus on the global entry policy setting, as solving for the optimal policy becomes more complex in the multi-seller case when multiple policy dimensions are involved. 
Our primary interest is to understand how sellers interact with each other under strategic platform entry.

\subsection{The Multi-Seller Markov Game}
We model the strategic interaction as a Stackelberg game, where the platform commits to a global entry policy $T_p$ and the sellers play a finite-horizon Markov Game $\mathcal{M}_{T_p}$ induced by $T_p$.

\subsubsection{The Sellers' Game.}
The game $\mathcal{M}_{T_p}$ is defined by:
\begin{enumerate}[leftmargin=*]
    \item \textbf{Agents}: 
    A set of $N$ sellers, indexed by $i \in \{1, 2, ..., N\}$, each with a discount factor of $\gamma_i$.
    \item \textbf{Products}: A set of $M$ products, indexed by $j \in \{1, 2, ..., M\}$.
    Each product has a known prior probability $p_j$ of being ``good'' (high reward $r^g_j$) or ``bad'' (low reward $r^b_j$). 
    Each seller $i$ has a one-time innovation cost of $c_{i, j}$ to explore product $j$.
    \item \textbf{State ($x_{t} \in \mathcal{X}$)}: The state at time $t$ includes 
        \begin{enumerate}
            \item The state of each product: \{undeveloped, good, bad, entered\}, denoted by $\{\sunexplore, \sgood, \sbad, \senter\}$.
            \item A $N \times M$ matrix indicating which sellers are currently offering which products.
            \item A $N \times M$ matrix tracking the time elapsed since each seller first offered each product.        \end{enumerate}
    \item \textbf{Actions ($a_t \in \mathcal{A}$)}: 
    The joint action $a_t = (a_{1,t}, \dots, a_{N,t})$ is the combination of individual seller actions, where $a_{i,t} \in \{0, 1, 2, ..., M\}$ represents seller $i$'s choice of product $j$ or nothing (i.e., 0).
    \item \textbf{Transitions ($x_{t+1} \sim \mathcal{T}(x_t, a_t)$)}:
    The state transitions based on the current state and the joint action. 
    Product states change upon seller exploration and platform entry $T_p$.
    %
    \item \textbf{Rewards ($r_{i, t}$)}:
    If seller $i$ chooses action $a_{i,t} = j$, their reward is:
    \[r_{i, t} = f_j(\hat{r}_j, n_{j,t}) - \mathcal{I}_{i, j} \cdot c_{i,j},\]
    where $\hat{r}_j$ is the realized reward of product $j$, $n_{j,t}$ is the number of sellers offering product $j$, $f_j$ is some function modeling the reward under different extent of market congestion, and $\mathcal{I}_{i, j}$ is an indicator variable that applies the cost on the first time seller $i$ offers product $j$. 
    If $a_{i,t} = 0$ or product $j$ has been entered by the platform, then $r_{i, t} = 0$.
    \item \textbf{Observations ($o_{i,t} \in \Omega$)}: The observation for seller $s_i$ at time $t$, $o_{i,t} \subset x_t$, contains every information from $x_t$ except for the private innovation cost of every other seller. 
\end{enumerate}

\subsubsection{Seller and Platform Objectives.}
Each seller $i$ chooses a policy $\pi_{i}$ to maximize their own total expected discounted reward. 
As each seller's reward depends on the actions of other sellers, the goal is to find a policy profile $\boldsymbol{\pi}^*(T_p) =(\pi^*_1, \dots, \pi^*_N)$ that is a Nash equilibrium
:
\[\pi_i^* \in \arg\max_{\pi_i} \mathbb{E}_{\pi_i, \pi_{-i}^*} \left[ \sum_{t=0}^T \gamma_i^t r_{i,t} \right], \quad \forall i \in \{1, ..., N\}.
\]


The platform's objective is to choose an entry time $T^*_p$ that maximizes its own utility, anticipating the sellers' equilibrium response $\pi^*(T_p)$. 
The platform's utility is the sum of discounted rewards from all products it has entered:
\[
u_P(T_p) = \mathbb{E}_{\pi^*(T_p)} \Bigg[\sum_{t=0}^T \sum_{j=1}^M \gamma_p^t \cdot r^g_j \cdot \mathcal{I}{\{x_{j, t} = \senter\}}\Bigg]
\]
The platform's optimization problem is thus:
\[
T_p^* = \arg\max_{T_p \ge 1} u_P(T_p).
\]

\subsection{Approximating a Seller Equilibrium}
The introduction of multiple sellers competing with each other creates a non-stationary environment, making analytical solutions (including state space and joint action) appear intractable. Therefore, we use deep reinforcement learning to model the sellers' strategic behavior. 
To address the non-stationarity of evolving seller strategies, we use an iterative best-response procedure to identify an approximate, $\epsilon$-Nash equilibrium:


\begin{enumerate}[leftmargin=*]
    \item Independent training: Sellers are trained in parallel for a set of $K$ episodes to learn policy profile, $\pi = (\pi_1, \dots, \pi_N)$.

    \item Iterative best-response: 
    For each agent, we evaluate the regret (or unilateral devation gain) by freezing the other agents' policies $\pi_{-i}$ and train a best-response policy $\pi_i^{\text{br}}$ against them, starting from initially trained $\pi_i$:
    \[
    \small
        \text{Regret}_i(\pi) = \max \left( 0, \quad \mathbb{E}_{\pi_i^{\text{br}}, \pi_{-i}} \left[ \sum_{t=0}^T \gamma^t r_{i,t} \right] - \mathbb{E}_{\pi_i, \pi_{-i}} \left[ \sum_{t=0}^T \gamma^t r_{i,t} \right] \right)
        \]

    \item Convergence check: 
    If the maximum regret across all agents falls below a threshold $\epsilon$, the policy profile is an approximate, $\epsilon$-Nash equilibrium, and the training terminates.
    Otherwise, we resume parallel training.

\end{enumerate}
While theoretical guarantees for MARL in general-sum games remain an open question, the iterative procedure enables the identification of empirically stable joint policies, corresponding to an approximate, $\epsilon$-Nash equilibrium.
To account for the possibility of multiple equilibria, we additionally execute this process with multiple random seeds to obtain a range of potential equilibrium outcomes for each $T_p$.

\section{A Simulation Study of Strategic Platform Entry in Multi-Seller Markets}
\label{sec:multiseller_exp}
We develop a {\em multi-agent Gym simulation environment} based on the multi-seller model in Section~\ref{sec:multi_seller_model} to 
examine how strategic platform entry shapes market outcomes under different seller competition structures.
%
The simulation allows us to explore a range of market configurations, facilitating counterfactual analysis and the interpretability of the market outcomes.

\subsection{Experiment Settings}

We construct simulation environments that capture key strategic forces that can influence platform entry and seller exploration, while avoiding excessive model complexity.
To this end, we focus on two canonical categories of market structures---\textit{clustered} and \textit{diverse}---which span the range of competitive conditions observed in real-world platforms. 
The clustered setting represents markets with a highly profitable or popular product space (e.g., sellers crowding into  categories like tech accessories), whereas the diverse setting captures markets characterized by differentiated niches (e.g., merchants specializing in  crafts like handmade jewelry or custom art). 

We adopt {\em two representative sellers}, capturing the simplest setting in which information spillover and competition can arise. 
Each seller can represent a group of similar sellers, allowing interpretable analysis of exploration and equilibrium under varying entry strategies with manageable computation.

\subsubsection{Market Environment Configurations.}
While we cannot use the Gittins index to derive a seller solution in settings with multiple sellers, we employ a Gittins-index-guided design to ensure that sellers' incentives align with the intended clustered or diverse market structures.
Specifically, for each sampled product reward profile, we adjust the corresponding innovation costs so that sellers' induced policies  generate the desired market scenario in the absence of platform entry (i.e., $T_p = \infty$).
For all scenarios, we introduce ``control products'' with Gittins indices below a fixed threshold, $\Bar{G}$, to serve as background alternatives, ensuring that observed exploration behaviors arise endogenously from incentive structures rather than from a lack of available options.
Below, we describe these distinct market environments and how they are generated. We denote $G_{i, j}$ as the formulation of the Gittins index of product $j$ for seller $i$.
 \begin{itemize}[leftmargin=*]
     \item \textit{Clustered Environments}\\
     There is a single popular product space $j^*$ whose shared-reward Gittins index exceeds the index of any other product $j$:
     \[G_{i, j^*}(\sunexplore; T_p=\infty) > \Bar{G} >G_{i, j}(\sunexplore;
     T_p=\infty), \quad \text{ $\forall$} i \text{ and } j\neq j^*.\]
     We consider two scenarios within this category, analyzing \textit{how cost structure and reward size shape seller incentives under clustered competition.}
     \begin{itemize}
         \item Scenario C1 (Standard): A high-demand product space. 
         \item Scenario C2 (High-stakes): A high-stakes product space with a large innovation cost but also a higher reward.
     \end{itemize}
     \item \textit{Diverse Environments}\\
        Sellers have specialized incentives, captured by their one-time innovation costs $c_{i,j}$, with each seller having a preferred product.
        We consider two scenarios, analyzing \textit{how asymmetries in seller capabilities shape strategic responses and product diversity.}
     \begin{itemize}
         \item Scenario D1 (Specialists):
         Each seller has a unique preferred product $j_i$, and faces prohibitively high costs to explore other’s niche:
         \[G_{i, j_i}(\sunexplore; T_p=\infty) > \Bar{G} > G_{i, j_{i'}}(\sunexplore; T_p=\infty), \quad \text{for } i \neq i' \text{ and } j_{i} \neq j_{i'}\]
         \item Scenario D2 (Specialist and Generalist): 
         This models a market with a specialist seller $i^*$ and a generalist seller $i$. 
         The generalist can enter the specialist's niche but still prefers their own product:
         \[G_{i^*, j_{i^*}}(\sunexplore; T_p=\infty) > \Bar{G} > G_{i^*, j_{i}}(\sunexplore; T_p=\infty) , \quad \text{for } j_{i^*} \neq j_{i}\]
         \[G_{i, j_{i}}(\sunexplore; T_p=\infty) > G_{i, j_{i^*}}(\sunexplore; T_p=\infty) > \Bar{G}, \quad \text{for } j_{i^*} \neq j_{i}\]
     \end{itemize}
 \end{itemize}

To cover a variety of risk-reward profiles, we sample product parameters from discrete sets, specifically 
$r^g_j \sim \{75, 100, 200\}$, $r^b_j \sim\{0, 25, 50\}$, and $p^g_j \sim \{0.2, 0.5, 0.8\}$.%
\footnote{In C2, we model high-stakes products by augmenting $r^g_j$ with an additional 500.}

\subsubsection{Agent Configuration.}
%
Given the large state space, we model each seller's exploration using {\em deep Q-learning (DQN)}, with best-response checks. 
Each seller trains its action–value function independently, and we iteratively compare policies to best responses against others’ fixed strategies to approximate an empirical $\epsilon$-Nash equilibrium, using a regret threshold of $\epsilon=0.33$ for convergence.

To identify the optimal platform entry time, we search over integer values of $T_p$. 
For each value, we run the MARL training procedure under multiple random seeds to find seller equilibria and compute the platform's expected revenue. 
We put all hyperparameter settings for training in Appendix \ref{app:hyperparameters}.

\subsection{Effect of Platform Entry: Empirical Analysis}
Figure~\ref{fig:clustered_results} presents our main findings on how platform entry affects different market structures, using several key metrics:
\begin{itemize}[leftmargin=*]
    \item \textit{Agent utilities}: total rewards for the platform, sellers, and consumers (i.e., rewards generated by the platform and sellers),
    \item \textit{Products explored}: the fraction of distinct products explored by sellers, measuring innovation,
    \item \textit{Product variety}: the fraction of distinct products offered per timestep, measuring product diversity,
    \item \textit{Cluster rate}: the frequency of sellers offering the same product.
\end{itemize}

For a given $T_p$ of an environment, we run a large number of simulations with different seeds on the resulting seller game under the trained DQN seller policy $\pi^*$ to capture these metrics.
In case when there are multiple equilibria, we report a range of values for these equilibrium outcome metrics.

\begin{figure*}
    \centering
    \begin{subfigure}[b]{0.43\textwidth}
        \centering

        \caption{Standard cluster (C1) Utility Metrics}
        \label{fig:c1_utility_results}
        \includegraphics[width=\textwidth]{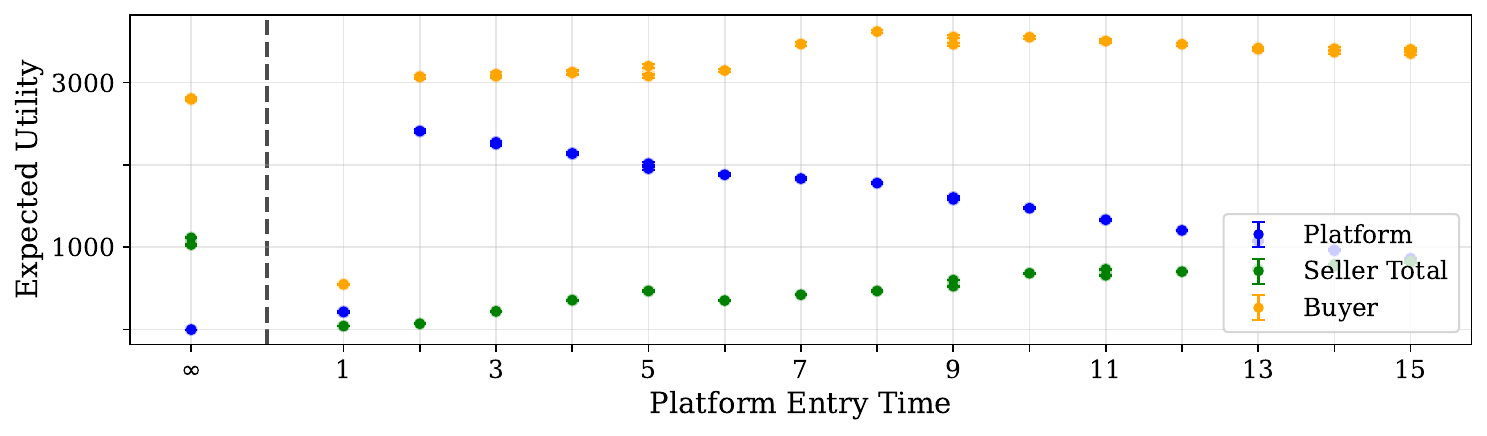}

    \end{subfigure}
    \hspace{0.75em}
    \begin{subfigure}[b]{0.43\textwidth}
        \centering
        \caption{Standard cluster (C1) Product Metrics}
        \label{fig:c1_product_results}
        \includegraphics[width=\textwidth]{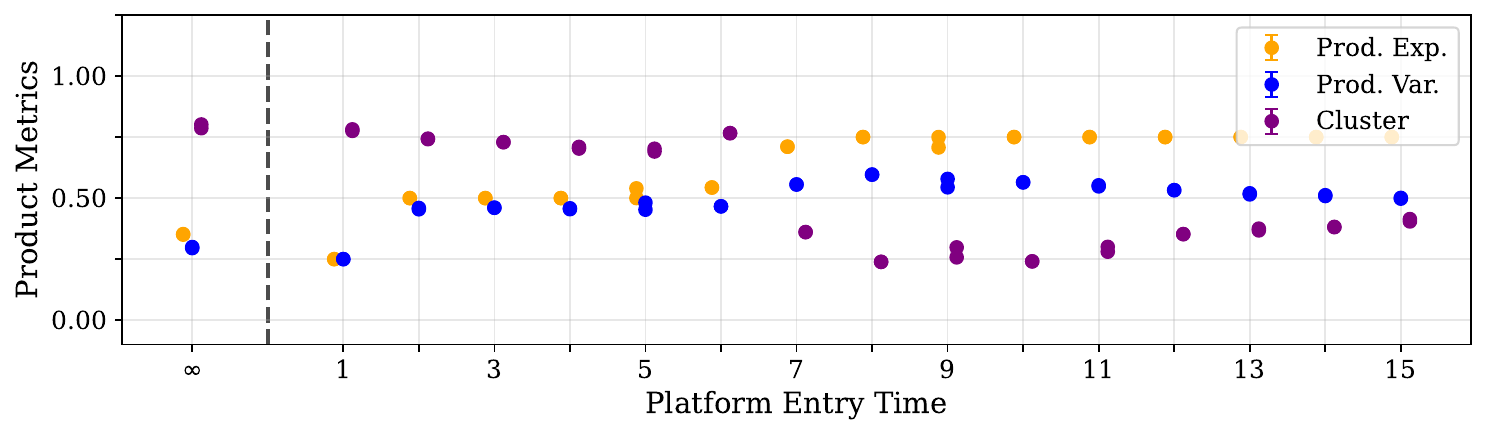}

    \end{subfigure}
    
    
    \begin{subfigure}[b]{0.43\textwidth}
        \centering
        \caption{High-stakes cluster (C2) Utility Metrics}
        \label{fig:c2_utility_results}
        \includegraphics[width=\textwidth]{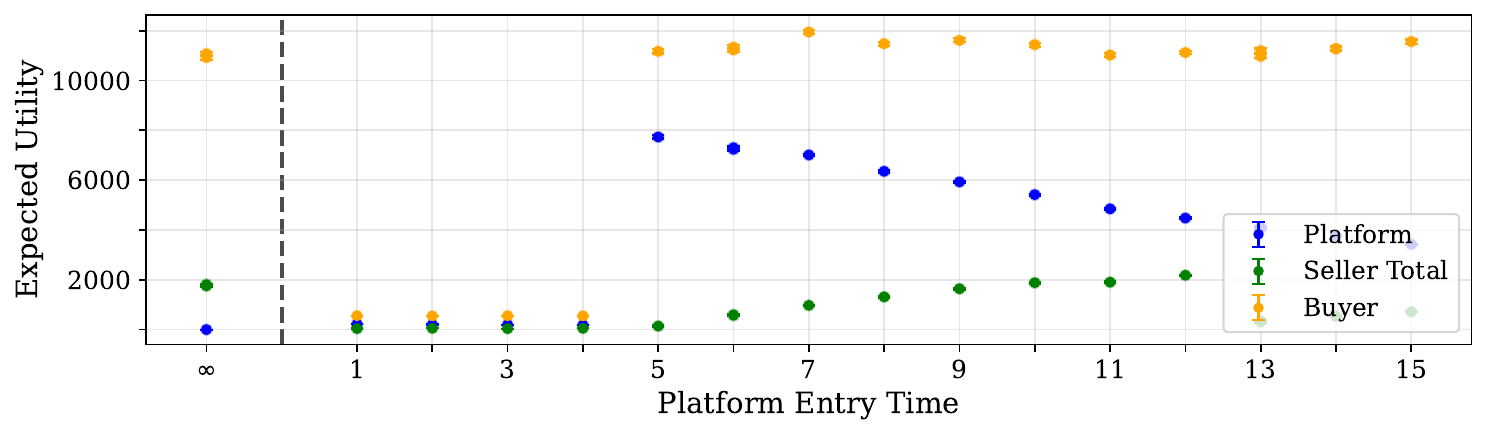}

    \end{subfigure}
    \hspace{0.75em}
    \begin{subfigure}[b]{0.43\textwidth}
        \centering
        \caption{High-stakes cluster (C2) Product Metrics}
        \label{fig:c2_product_results}
        \includegraphics[width=\textwidth]{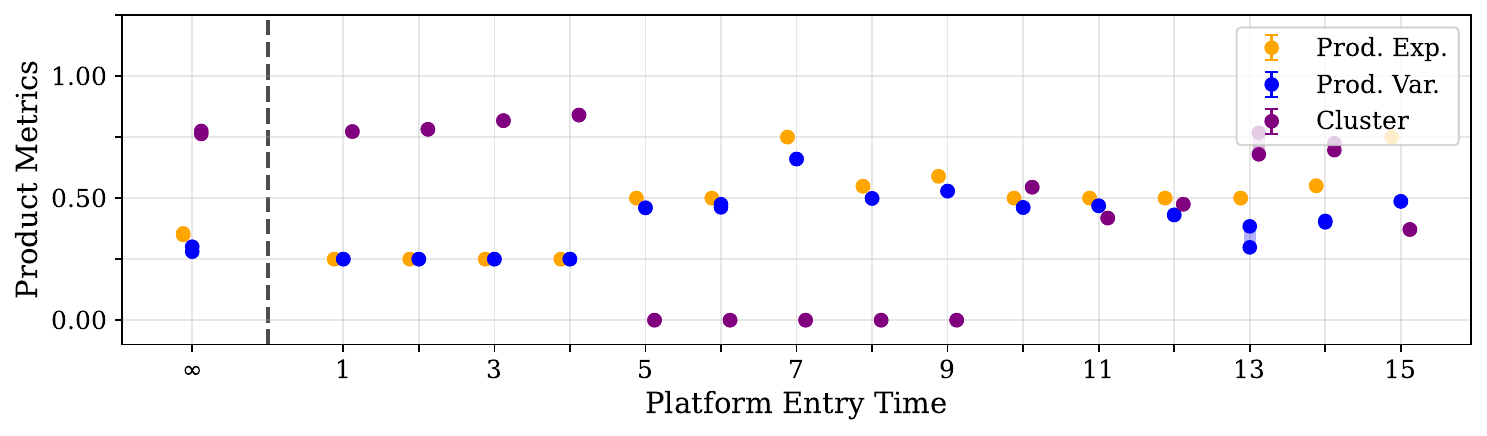}

    \end{subfigure}


    \begin{subfigure}[b]{0.43\textwidth}
        \centering
        \caption{Diverse specialists (D1) Utility Metrics}
        \label{fig:d1_utility_results}
        \includegraphics[width=\textwidth]{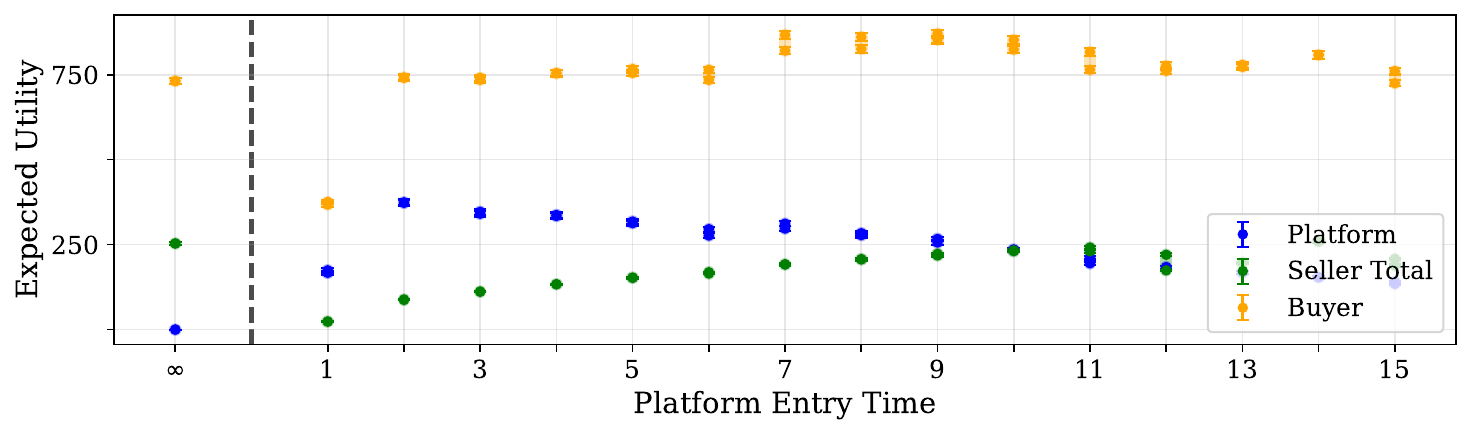}

    \end{subfigure}
    \hspace{0.75em}
    \begin{subfigure}[b]{0.43\textwidth}
        \centering
        \caption{Diverse specialists (D1) Product metrics}
        \label{fig:d1_product_results}
        \includegraphics[width=\textwidth]{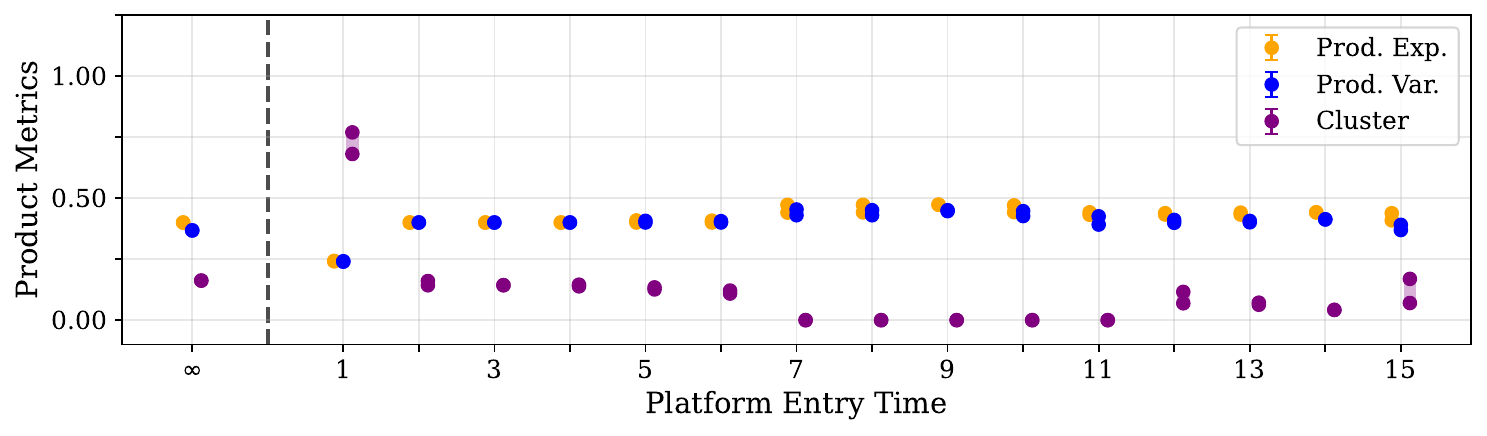}

    \end{subfigure}
    
    
    \begin{subfigure}[b]{0.43\textwidth}
        \centering
        \caption{Diverse specialist and generalist (D2) Utility Metrics}
        \label{fig:d2_utility_results}
        \includegraphics[width=\textwidth]{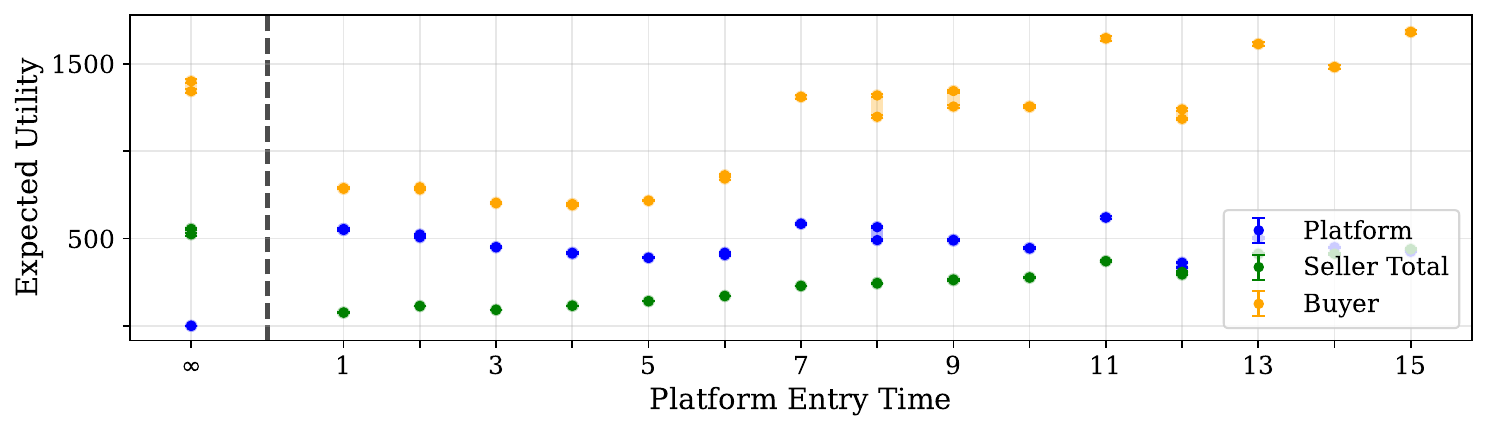}

    \end{subfigure}
    \hspace{0.75em}
    \begin{subfigure}[b]{0.43\textwidth}
        \centering
        \caption{Diverse specialist and generalist (D2) Product Metrics}
        \label{fig:d2_product_results}
        \includegraphics[width=\textwidth]{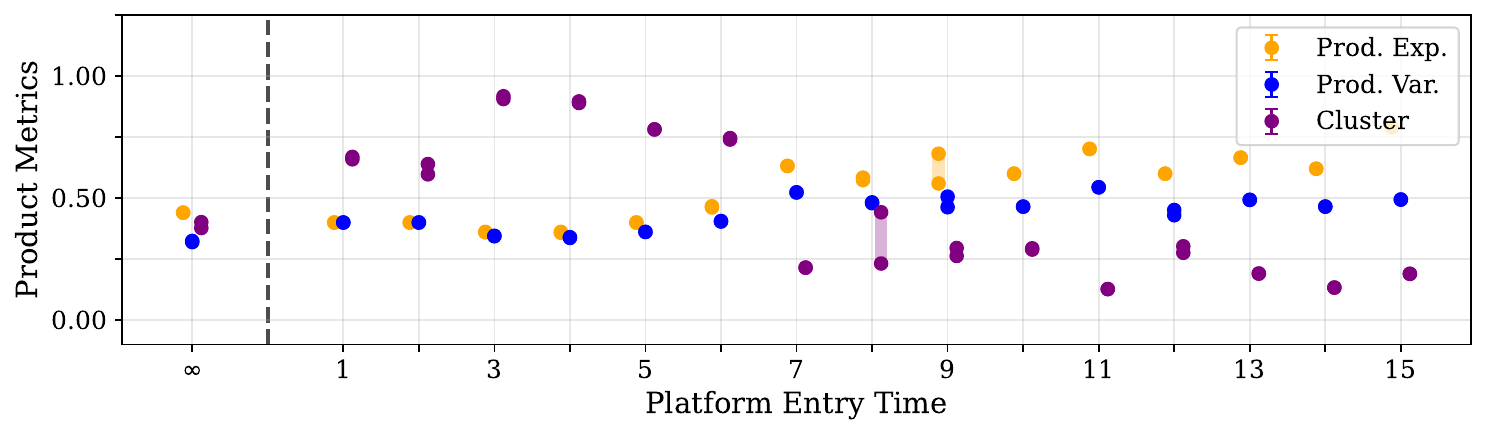}

    \end{subfigure}
    \vspace{-1.5ex}
    \caption{Expected utility metrics for the platform (blue), sellers (green), and buyers (yellow) as a function of $T_p$ on left. Products explored (yellow), product variety (blue), and cluster rate (purple) as a function of $T_p$ on right. These metrics are based on the results of 4,000 environment simulations. The range at each $T_p$ denotes outcomes from multiple equilibria. 
    In C1, optimal entry is early, $T^*_p \approx 2$, to capture value from the clustered product. 
    In C2, optimal entry is delayed, $T^*_p \approx 5$, to cover high innovation costs. 
    In D1, optimal entry, $T^*_p \approx 2$,  has a muted effect on innovation. In D2, optimal entry is delayed, $T^*_p \approx 11$. 
}
    \label{fig:clustered_results}
\end{figure*}

\subsubsection{Clustered Environments.}

 We see that a well-timed platform entry can promote seller diversification: the threat of entry on a high-demand product encourages sellers to explore alternatives.
This echoes empirical evidence~\citep{Wen19}, showing that app developers proactively shift innovation by improving their current apps or developing new apps before platform entry occurs (e.g., moving from general health apps to niche ski apps). 

In the high-stake cluster scenario C2
(Fig.~\ref{fig:c2_utility_results}),
the platform is incentivized to delay its entry until $T^*_p = 5$, capturing value from one seller exploring the high-demand product while prompting the other to innovate elsewhere. Here, the platform's objective aligns more closely with seller exploration/diversification and social welfare, as sellers use the longer protection period to justify higher costs. Early entry can discourage the development of potentially lucrative products, limiting the platform's value capture. 

In contrast, in the standard cluster scenario C1 (Fig.~\ref{fig:c1_utility_results}), the optimal entry occurs earlier ($T^*_{p} = 2$) where the sellers still cluster despite early entry. 
This aggressive entry does not fully align with social welfare or market diversity goals, as the market favors a later entry ($T^*_{sw} = 8$) to preserve seller incentives to explore riskier or lower-demand products, highlighting the potential need for regulatory intervention to discourage overly aggressive platform entry.







\subsubsection{Diverse Environments.}

Overall, early platform entry either reduces market diversity (in D2) or has little impact relative to no-entry (in D1), whereas delayed entry increases market diversity.

This effect is most pronounced in the specialist-generalist scenario D2 (Fig. \ref{fig:d2_product_results}), which mirrors real-world online markets that feature a mix of niche innovators (specialists) and established sellers (generalists) that can pivot across product spaces.
When facing early entry, the generalist seller may choose to sell clustered products preemptively, capturing short-term gains before imitation occurs, since their flexibility allows them to explore alternative products later (Fig.~\ref{fig:d2_product_results}). This mirrors  empirical evidence ~\citep{Wen19}, which shows that app developers who have a portfolio of apps unaffected by entry can shift to existing apps more easily. 
This reduces overall product exploration, diversity, and buyer utility. 
Consequently, the platform's optimal entry time is much later ($T^*_p=11$), providing a longer protection period that encourages sellers to pursue niche space and maintain market diversity (Fig.~\ref{fig:d2_utility_results}).

In D1 (specialists) markets, by contrast, sellers occupy distinct niches. 
Platform entry can occur earlier without significant disruption, mainly to extract value rather than reshape seller behavior (Fig.~\ref{fig:d1_product_results}).

\vspace{-1ex}
\section{Discussion}
\label{sec:discussion}

The optimal single-seller policy for a given platform policy $\boldsymbol{\pi_p}$, can be derived using a closed-form Gittins index. 
Given this, we solve for the optimal platform policy by optimizing over a finite set of Pareto-optimal points within regions corresponding to distinct seller strategies. 
In multi-seller settings, we use deep reinforcement learning to train seller policies and solve for approximately optimal platform entry under approximate seller equilibria.

In the single-seller model, we explore different kinds of platform policies and consider different market compositions. Higher transaction fees are observed in markets with predictable demand and low innovation costs, while lower fees are observed in innovation-driven markets with uncertain demand. However, excessive fees reduce seller profit and slow innovation, making caps on fees effective in restoring seller profits while increasing buyer utility. 
Heterogeneous platform entry across products boosts buyer utility through flexible entry-timing but can hurt seller profits. 
Imposing limits on the earliest possible platform entry balances these effects. 

With multiple sellers, we explore seller-seller and platform-sellers interaction in simulation, across different, clustered and diverse, market environments. Platforms tend to enter aggressively when products are in high demand and offered by many sellers, but choose to commit to delayed entry for products with high costs or uncertain outcomes. This aligns with findings from \citet{Zhu18} that show Amazon enters categories like toys and games but avoids high-cost categories.
More interestingly, sellers adapt even before entry occurs. Platform entry can disrupt clustered product markets, prompting some sellers to explore alternatives. Under aggressive platform entry, more versatile sellers may preemptively cluster into other sellers' products to capture short-term gains before entry. 

Overall, our findings indicate that market structure plays a large role in determining whether platform entry aligns with seller exploration. In settings where alignment occurs, social welfare, innovation, and market diversity increases. When market structure favors aggressive entry or the platform enters early timing, seller exploration can be reduced compared to social-welfare optimal entry, and regulatory interventions may be needed to restore social welfare.
\bibliographystyle{ACM-Reference-Format}
\bibliography{refs}

\appendix
\section{Gittins Index}
\label{app:gindex}
The optimality of the Gittins-index policy holds under the assumptions that each arm is an independent stochastic process with discounted rewards and is stationary, i.e., playing one arm does not influence the state or rewards of another and the state of an unplayed arm does not change over time~\citep{weber-proof}. In our setting, each ``arm" corresponds to a different product.

The index, denoted $G_{j}(x_{j})$, for an arm $j$ in a given state $x_j$, is defined as the maximal expected reward rate, where the maximization is over all possible future stopping times $\tau \geq 1$:

\[
    G_{j}(x_{j}) = \sup_{\tau \ge 1} \frac{\mathbb{E}\left[\sum_{t=0}^{\tau-1} \gamma^t R_{j}(x_{j}(t)) \mid x_{j}(0) = x_{j}\right]}{\mathbb{E}\left[\sum_{t=0}^{\tau-1} \gamma^t \mid x_{j}(0) = x_{j}\right]},
\]
where $\gamma$ is the discount factor,  $R_j(x_j(t))$ is the reward from arm $j$ at time $t$, and the expectation $\mathbb{E}[\cdot]$ is taken over the stochastic evolution of the arm's state.
\section{Deferred Analysis for Section~\ref{sec:single_seller_model}}
\subsection{Gittins Index Calculation}
\label{app:gindex_derivation_single_platform}
We derive the Gittins index under each stopping rule. We generalize the derivation to include transaction fee $\alpha$ and for $T_{p_j} = T_p$ or $T_{p_j} \in \mathbf{T_p}$. In the global and heterogeneous entry case, set $\alpha = 0$. We use $V$ to denote the total expected discounted reward under a certain stopping rule, and $D$ to denote the corresponding total expected discounted time horizon. 
\\
Under $\tau_1$, we have:
\[G_j^{(1)}(\sunexplore; T_{p_j}, \alpha)=0\]
Under $\tau_2$, we have:
\[
\begin{aligned}
V^{(2)} &= -c_j + p_j \sum_{t=0}^{T_{p_j}-1}\gamma^t (1 - \alpha)r^g_j + (1-p_j)\sum_{t=0}^{\infty}\gamma^t (1 - \alpha)r^b_j \\
&= -c_j + (1-\alpha)\bigg( p_j r^g_j \frac{1-\gamma^{T_{p_j}}}{1-\gamma} + (1-p_j) \frac{r^b_j}{1-\gamma}\bigg) \\
D^{(2)} &= p_j \sum_{t=0}^{T_{p_j}-1}\gamma^t + (1-p_j)\sum_{t=0}^{\infty}\gamma^t = p_j \frac{1-\gamma^{T_{p_j}}}{1-\gamma} + \frac{1-p_j}{1-\gamma} \\
G_j^{(2)}(\sunexplore; T_{p_j}) &= \frac{(1-\gamma)(-c_j) + (1-\alpha)\big( p_j r^g_j(1-\gamma^{T_{p_j}}) + (1-p_j)r^b_j\big)}{p_j(1-\gamma^{T_{p_j}}) + (1-p_j)}
\end{aligned}
\]
Under $\tau_3$, we have:
\[
\begin{aligned}
V^{(3)} &= -c_j + p_j \sum_{t=0}^{T_{p_j}-1}\gamma^t (1-\alpha)r^g_j + (1-p_j)(1-\alpha)r^b_j
\\
&= -c_j + (1-\alpha)\big(p_j r^g_j\frac{1-\gamma^{T_{p_j}}}{1-\gamma} + (1-p_j)r^b_j\big)
\\
D^{(3)} &= p_j \sum_{t=0}^{T_{p_j}-1}\gamma^t + (1-p_j)(\gamma^0) = p_j \frac{1-\gamma^{T_{p_j}}}{1-\gamma} + (1-p_j)\\
G_j^{(3)}(\sunexplore; T_{p_j}) &= \frac{-c_j + (1-\alpha)\big(p_j r^g_j\frac{1-\gamma^{T_{p_j}}}{1-\gamma} + (1-p_j)r^b_j\big)}{p_j \frac{1-\gamma^{T_{p_j}}}{1-\gamma} + (1-p_j)}
\end{aligned}
\]
The true Gittins index is the maximum value across all stopping rules. 
For our model, this is the maximum of the indices derived from these candidate rules.
\[
    G_j(\sunexplore; T_{p_j}) = \max \left\{ 0, G_j^{(2)}(\sunexplore; T_{p_j}), G_j^{(3)}(\sunexplore; T_{p_j}) \right\}
\]
$G_j(G; T_{p_j}) = r^g_j$, $G_j(B; T_{p_j}) = r^b_j$, and $G_j(E; T_{p_j}) = 0$ due to the persistence of the reward.

\subsection{Optimal Platform Policy Algorithm}
\label{app:platpolalg}
\begin{algorithm}[H]
\caption{Optimizing platform policy $\boldsymbol{\pi_p}$}
\begin{algorithmic}[1]
\STATE \textbf{Input:} Product profiles of $M$ products
\STATE Generate boundary set $\mathcal{B}$
\STATE Generate regions $\mathcal{R}(\mathcal{B})$
\FOR{each region $R_i$ in $\mathcal{R}(\mathcal{B})$}
    \STATE Find Pareto optimal points in $R_i$
    \STATE Select point $\boldsymbol{\pi_p^{i}}$ that maximizes $u_p(\boldsymbol{\pi_p^i 
    |R_i)}$
\ENDFOR
\STATE \textbf{Return} policy $\boldsymbol{\pi_p^*}$ among $\boldsymbol{\pi_p^{i}}$ with highest utility
\end{algorithmic}
\end{algorithm}

\subsection{Global $T_p$ Toy Example}
\label{app:globaltoyex}
\begin{figure}[H]
    \centering
    \begin{subfigure}{\columnwidth}	
	\centering
	\includegraphics[width=\columnwidth]{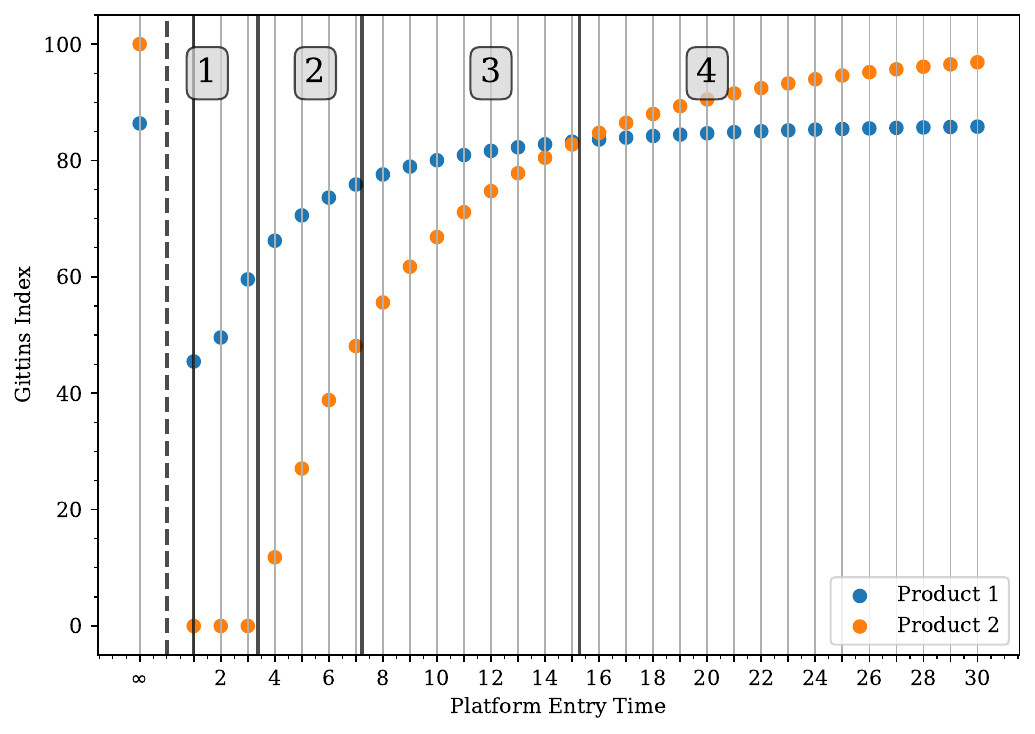}
        \vspace{-4ex}
    \end{subfigure}
    \caption{Gittins indices for Product Type A and Product Type B as a function of global $T_p$.}
    \label{fig:gittins_tp}
\end{figure}
We consider a two-product environment where product 1 is of Type A and product 2 is of Type B, as defined in Table \ref{tab:ProductTypes}. The Gittins index of product 1 and product 2 are plotted as a function of $T_p$ in Fig. \ref{fig:gittins_tp}. 

The zero boundary of product 1 is $b_{1}^0 = 1$ and the zero boundary of product 2 is $b_{2}^0 = 3.38$. The bad-indifference boundary of product 1 is $b_{1, 2}^B = 7.23$ and there is no bad-indifference boundary for product 2. The undeveloped-indifference boundary $b^U_{1, 2} = 15.27$. This gives us 4 regions, $R_1=(1, 3.38)$, $R_2=(3.38, 7.23)$, $R_3=(7.23, 15.27$), and $R_4=(15.27, \infty)$. The Pareto optimal points are 1, 4, 8, and 16 for the four regions respectively. The optimal $T_p^*$ is $T_p^* =8$.

We also characterize the seller strategy in the following regions: 
\begin{itemize}[leftmargin=*]
    \item Region 1: When $T_p$ is small, the seller will only explore A. The threat of immediate platform entry makes the riskier product B unattractive.
    \item Region 2: As $T_p$ increases, the seller will explore A first. 
    If A transitions to the good state, the seller will explore B after $T_p$ steps; otherwise, the seller will remain selling A in its bad state.
    \item Region 3: As $T_p$ increases even more (i.e., $T_p \ge 8$), the seller still explores A first, and the seller will choose to explore B immediately if A ends in the bad state.
    \item Region 4: As $T_p$ becomes large (i.e., $T_p \ge 16$), the seller explores Product B first. 
    If B transitions to the good state, the seller will explore A after $T_p$; otherwise, the seller immediately switches to A.
    Note this is the same behavior when there is no platform entry.
\end{itemize}

\subsection{Heterogeneous $\mathbf{T_p}$ Toy Example}
\label{app:heterotoyex}
\begin{figure}[H]
    \centering
    \begin{subfigure}{\columnwidth}	
	\centering	\includegraphics[width=\columnwidth]{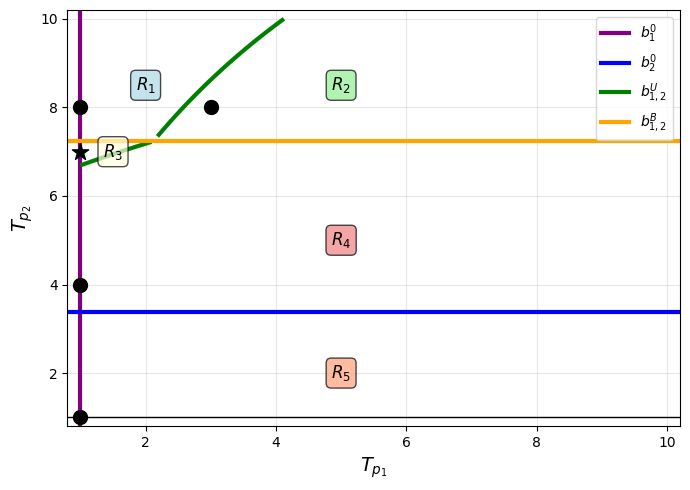}
    \end{subfigure}
    \caption{Boundaries and regions induced by Product Type A and Product Type B for heterogeneous $\mathbf{T_p}$}
    \label{fig:hetero_tp}
\end{figure}

We consider a two-product environment where product 1 is of Type A and product 2 is of Type B, as defined in Table \ref{tab:ProductTypes}. The boundaries and regions are labeled in Fig. \ref{fig:hetero_tp}. We do not need to optimize over the entire policy space for each $R_i$, but focus only on the Pareto optimal points where the platform's utility $u_p$ is monotonic decreasing with respect to $T_{p_1}$ and $T_{p_2}$. The optimal heterogeneous entry is $\mathbf{T_p^*} = (1, 7)$. 

\section{Multi-seller Environment}
\subsection{Hyperparameters }
\label{app:hyperparameters}
\renewcommand{\arraystretch}{0.9} 
\begin{table}[H]
\centering
\small

\begin{minipage}[t]{0.48\columnwidth}
\centering
\textbf{IQL - DQN} \\
\begin{tabular}{@{}lr@{}}
\toprule
Param & Value \\
\midrule
LR & 0.0001 \\
Discount factor & 0.9 \\
Batch Size & 32 \\
Buffer Size & 450000 states \\
Exploration Decay & 0.99995 \\
Exploration Factor & 1 $\rightarrow$ 0.25 \\
\bottomrule
\end{tabular}
\end{minipage}%
\hspace*{0.04\columnwidth}%
\begin{minipage}[t]{0.48\columnwidth}
\centering
\textbf{Iterative Best-Response} \\
\begin{tabular}{@{}lr@{}}
\toprule
Param & Value \\
\midrule
Exploration Decay & 0.999925 \\
Exploration Factor & 1 $\rightarrow$ 0.1 \\
Total Eps & 70000 \\
$\epsilon$ (for converge) & 0.33 \\
\bottomrule
\end{tabular}
\end{minipage}

\caption*{Hyperparameters for Training: All hyperparameters omitted from Iterative Best-Response share the same hyperparameters as IQL-DQN.}
\end{table}

\end{document}

%% file: macros.tex
\newcommand{\sunexplore}{U}
\newcommand{\sgood}{G}
\newcommand{\sbad}{B}
\newcommand{\senter}{E}